\documentclass[a4paper,english,nolineno]{socg-lipics-v2021}
\newcommand{\PAPER}[1]{#1}
\newcommand{\SOCG}[1]{}

\usepackage[utf8]{inputenc}

\usepackage{microtype,hyperref}
   \hypersetup{%
      breaklinks,
      colorlinks=true,%
      urlcolor=[rgb]{0.25,0.0,0.0},%
      linkcolor=[rgb]{0.5,0.0,0.0},%
      citecolor=[rgb]{0,0.2,0.445},%
      filecolor=[rgb]{0,0,0.4},
      anchorcolor=[rgb]={0.0,0.1,0.2}%
   }

\usepackage{amsmath,amssymb,amsthm,amsfonts,dsfont,mathtools}
\usepackage{xspace}
\usepackage{algorithm}
\usepackage[noend]{algpseudocode}

\algnewcommand{\Inputs}[1]{%
  \State \textbf{Inputs:}
  \Statex \hspace*{\algorithmicindent}\parbox[t]{.8\linewidth}{\raggedright #1}
}
\algnewcommand{\Initialize}[1]{%
  \State \textbf{Initialize:}
  \Statex \hspace*{\algorithmicindent}\parbox[t]{.8\linewidth}{\raggedright #1}
}
\algnewcommand{\TurnOne}[1]{%
  \State \textbf{Timestep 1:}
  \Statex \hspace*{\algorithmicindent}\parbox[t]{.8\linewidth}{\raggedright #1}
}

\makeatletter
\newsavebox{\@brx}
\newcommand{\llangle}[1][]{\savebox{\@brx}{\(\m@th{#1\langle}\)}%
  \mathopen{\copy\@brx\mkern2mu\kern-0.9\wd\@brx\usebox{\@brx}}}
\newcommand{\rrangle}[1][]{\savebox{\@brx}{\(\m@th{#1\rangle}\)}%
  \mathclose{\copy\@brx\mkern2mu\kern-0.9\wd\@brx\usebox{\@brx}}}
\makeatother

\newcommand{\ignore}[1]{}

\newcommand{\OPT}{\mbox{\rm OPT}}

\newcommand{\polylog}{\mathop{\mathrm{polylog}}}

\newcommand{\eps}{\varepsilon}
\newcommand{\OO}{\tilde{O}}

\makeatletter
\DeclareRobustCommand\onedot{\futurelet\@let@token\@onedot}
\def\@onedot{\ifx\@let@token.\else.\null\fi\xspace}

\makeatother

\title{More Dynamic Data Structures for Geometric Set Cover with Sublinear Update Time}

\author{Timothy M. Chan}{Department of Computer Science, University of Illinois at Urbana-Champaign, USA}{tmc@illinois.edu}{https://orcid.org/0000-0002-8093-0675}{
Supported in part by NSF Grant CCF-1814026.}

\author{Qizheng He}{Department of Computer Science, University of Illinois at Urbana-Champaign, USA}{qizheng6@illinois.edu}{}{}

\titlerunning{More Dynamic Data Structures for Geometric Set Cover}
\authorrunning{T.\,M. Chan and Q. He}
\Copyright{Timothy M. Chan and Qizheng He}



\hideLIPIcs  

\EventEditors{Kevin Buchin and \'{E}ric Colin de Verdi\`{e}re}
\EventNoEds{2}
\EventLongTitle{37th International Symposium on Computational Geometry (SoCG 2021)}
\EventShortTitle{SoCG 2021}
\EventAcronym{SoCG}
\EventYear{2021}
\EventDate{June 7--11, 2021}
\EventLocation{Buffalo, NY, USA}
\EventLogo{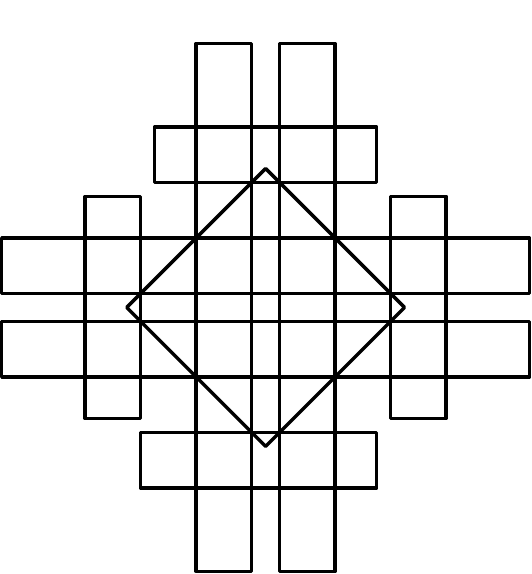}
\SeriesVolume{189}
\ArticleNo{53}

\begin{document}
\maketitle

\begin{abstract}
We study geometric set cover problems in dynamic settings, allowing insertions and deletions of points and objects. We present the first dynamic data structure that can maintain an $O(1)$-approximation in sublinear update time for set cover for axis-aligned squares in 2D\@. More precisely, we obtain randomized update time $O(n^{2/3+\delta})$ for an arbitrarily small constant $\delta>0$. Previously, a dynamic geometric set cover data structure with sublinear update time was known only for unit squares by Agarwal, Chang, Suri, Xiao, and Xue [SoCG 2020]. If only an approximate size of the solution is needed, then
we can also obtain sublinear amortized update time for disks in 2D and halfspaces in 3D\@. As a byproduct, our techniques for dynamic set cover also yield an optimal randomized $O(n\log n)$-time algorithm for static set cover for 2D disks and 3D halfspaces, improving our earlier $O(n\log n(\log\log n)^{O(1)})$ result [SoCG 2020].
\end{abstract}

\section{Introduction}

Approximation algorithms for NP-hard problems and dynamic data structures are two of the major themes studied by the algorithms community.  Recently, problems at the intersection of these two threads have gained much attention, and researchers in computational geometry have also started to systematically explore such problems.  For example, at SoCG last year, two papers appeared, one on dynamic geometric set cover by Agarwal et al.~\cite{agarwal2020dynamic}, and another on dynamic geometric independent set by
Henzinger, Neumann, and Wiese~\cite{DBLP:conf/compgeom/Henzinger0W20}.
In this paper, we continue the study by Agarwal et al.~\cite{agarwal2020dynamic} and
investigate dynamic data structures for approximating the minimum set cover in natural geometric instances.

\subparagraph*{Static geometric set cover.} 
In the static (unweighted) \emph{geometric set cover} problem,
we are given a set $X$ of $O(n)$ points and a set $S$ of $O(n)$ geometric objects,  and we want to find the smallest subset of objects in $S$ that covers all points of $X$.  The problem is fundamental and has many applications.  Let $\OPT$ denote the value (i.e., cardinality) of the optimal solution.
As the problem is NP-hard for many classes of geometric objects, we are interested in efficient approximation algorithms.

Many classes of objects, such as squares and disks in 2D, objects with linear union complexity, and halfspaces in 3D, admit polynomial-time $O(1)$-approximation algorithms (i.e., computing a solution of size $O(1)\cdot \OPT$), by using $\eps$-nets and LP rounding or the multiplicative weight update (MWU) method~\cite{bronnimann1995almost,clarkson1993algorithms,clarkson2007improved}.  Some classes of objects, such as disks in 2D and halfspaces in 3D, even have PTASs~\cite{mustafa2009ptas} or quasi-PTASs~\cite{MustafaRR15}, though with large polynomial running time.

Agarwal and Pan~\cite{agarwal2014near} worked towards finding \emph{faster} approximation algorithms that run in near linear time.  They gave randomized $O(n\log^4n)$-time algorithms (based on MWU) for computing $O(1)$-approximations, for example, for disks in 2D and halfspaces in 3D\@.  At last year's SoCG~\cite{ChanH20}, we described further improvements to the running time for 2D disks and
3D halfspaces, including a deterministic $O(n\log^3 n\log\log n)$-time algorithm and a randomized $O(n\log n (\log\log n)^{O(1)})$-time algorithm.

\subparagraph*{Dynamic geometric set cover.}
It is natural to consider the \emph{dynamic} setting of the geometric set cover problem.  Here, we want to support insertions and deletions of points in $X$ \emph{as well as} insertions and deletions of objects in $S$, while maintaining an approximate solution.  
Note that the solution may have linear size in the worst case. In the simplest version of the problem, we may just want to output the value of the solution.  More strongly, we may want some representation of the solution itself, so that afterwards, the objects in the solution can be reported in constant time per element when needed.

Agarwal et al.~\cite{agarwal2020dynamic} gave a number of results on dynamic geometric set cover.  They showed that for intervals in 1D, a $(1+\eps)$-approximation can be  maintained in $O((1/\eps)n^\delta)$ time per insertions and deletions of points and intervals.  (Throughout the paper, $\delta>0$ denotes an arbitrarily small constant.)
In 2D, they had only one main result: a fully dynamic  $O(1)$-approximation algorithm for unit axis-aligned squares, with
$O(n^{1/2+\delta})$ update time.

\subparagraph*{New results.}
We present several new results on dynamic geometric set cover.
The first is a fully dynamic, randomized, $O(1)$-approximation algorithm for the more general case of arbitrary axis-aligned squares in 2D, with $O(n^{2/3+\delta})$ amortized update time.
Though our time bound is a little worse than Agarwal et al.'s for the unit square case, the arbitrary square case is more challenging.  (The unit square case reduces to case of dominance ranges, i.e., quadrants, via a standard grid approach; since the union of such ranges forms a ``staircase'' sequence of vertices, the problem is in some sense ``$1.5$-dimensional''.  In contrast, the arbitrary square problem requires truly ``2-dimensional'' ideas.)

We then consider the case of halfspaces in 3D\@.  This case is fundamental, as set cover for 2D disks reduces to set cover for 3D halfspaces by the standard lifting transformation~\cite{BerBOOK}.  Also, by duality, hitting set for 3D halfspaces is equivalent to set cover for 3D halfspaces, and
hitting set for 2D disks reduces to set cover for 3D  halfspaces as well.

For 3D halfspaces, we obtain a fully dynamic, randomized algorithm with $O(n^{12/13+\delta})\le O(n^{0.924})$ amortized update time.  
Our result here is slightly weaker: it only finds the value of an $O(1)$-approximate solution (which could be good enough in some applications).  If a solution itself is required, we can still get sublinear update time as long as $\OPT$ is sublinear (below $n^{1-\delta}$).  This assumption seems reasonable, since sublinear reporting time is not possible otherwise.  (However, we currently do not know how to obtain a stronger time bound of the form $O(n^{\alpha}+\OPT)$  with $\alpha<1$ to report a solution for 3D halfspaces.)

\subparagraph*{Remarks.}
Our results are randomized in the Monte Carlo sense: the computed solution may not always be correct, but it is correct with high probability (w.h.p.), i.e., probability at least $1-1/n^c$ for an arbitrarily large constant~$c$.  The error probability bounds hold even when the user has knowledge of the random choices made by the algorithms.  (We do not assume an ``oblivious adversary''; in fact, the algorithms make a new set of random choices each time it computes a solution, and the data structures themselves are not randomized.)

The update times are better than stated in many cases, notably, when $\OPT$ is small or when $\OPT$ is large.  More precise $\OPT$-sensitive bounds are given in lemmas and theorems throughout the paper.  (The $O(n^{2/3+\delta})$ and $O(n^{12/13+\delta})$ bounds are obtained by ``balancing''.)

We have assumed that we are required to compute a solution after every update.
In some (but not all) cases, the update cost is smaller than the cost of computing a solution (this may be useful if we are executing a batch of updates).

\subparagraph*{Techniques.}
Our algorithms are obtained by handling two cases differently: when $\OPT$ is small and when $\OPT$ is large.
Intuitively, the small $\OPT$ case is easier since we are generating fewer objects, but the large $\OPT$ case also seems potentially easier since we can tolerate a larger additive error when targeting an $O(1)$-factor approximation---so, we are in a ``win-win'' situation.  For our algorithms for  3D halfspaces, we even find it necessary to
handle an intermediate case when $\OPT$ is medium (aiming for sublinear time for sublinear $\OPT$).

Our algorithms for the small $\OPT$ case are based on the previous static MWU algorithms~\cite{bronnimann1995almost,agarwal2014near,ChanH20}.  
The adaptation of these static algorithms is not straightforward, and requires using various known techniques in new ways ($(\le k)$-levels in arrangements, for our 2D square algorithm in Section~\ref{square_algo:small}, and ``augmented'' partition trees, for our 3D halfspace algorithm in Section~\ref{algo:small}).  The medium case for 3D halfspaces (in Section~\ref{app:medium}) is technically even more challenging (where we use ``shallow'' partition trees and other ideas).

Our algorithm for squares in the large $\OPT$ case (in Section~\ref{square_algo:large}) is different (not based on MWU), and interestingly uses quadtrees in a non-obvious way.  For 3D halfspaces, our algorithm in the large $\OPT$ case (in Section~\ref{algo:large}) can compute only the value of an approximate solution, and is based on random sampling.  However, the obvious way to use a random sample (just solving the problem on a random subset of points and objects) does not work.  We use sampling in a nontrivial way, combining geometric cuttings with planar graph separators.

In some of our algorithms, notably the small $\OPT$ algorithms for squares and 3D halfspaces (in Sections \ref{square_algo:small} and \ref{algo:small}) and the large $\OPT$ algorithm for 3D halfspaces (in Section~\ref{algo:large}), the data structure part is ``minimal'': we just assume that points and objects are stored separately in standard range searching data structures.  We describe sublinear-time algorithms to compute a solution from scratch, using range searching as oracles.  Dynamization becomes trivial, since range searching data structures typically are already known to support insertions and deletions.  (It also potentially enables other operations like merging sets, or solving the set cover problem for range-restricted subsets of points and subsets of objects.)

The topic of \emph{sublinear-time algorithms} has received considerable attention in the algorithms community, due to applications to big data (where we want to solve problems without examining the entire input).  A similar model of sublinear-time algorithms where the input is augmented with range searching data structures was proposed by Czumaj et al.~\cite{CzumajEFMNRS05}, who presented results on approximating the weight of the Euclidean minimum spanning tree in any constant dimension under this model.

\subparagraph*{Application to static geometric set cover.}
Although we did not intend to revisit the static problem,  our techniques can lead a randomized $O(1)$-approximation algorithm for set cover for 3D halfspaces running in $O(n\log n)$ time, which completely eliminates the extra $\log\log n$ factors in our previous result from SoCG'20~\cite{ChanH20} and is optimal (in comparison-based models)!  This bonus result is interesting in its own right, and is described in Section~\ref{sec:static}.



\PAPER{\section{Review of an MWU Algorithm}\label{sec:mwu}}
\SOCG{\section{Review of an MWU algorithm}\label{sec:mwu}}

We begin by briefly reviewing a known static approximation algorithm
for geometric set cover, based on the \emph{multiplicative weight updates (MWU)} method.
Some of our dynamic algorithms will be built upon this algorithm.   

Specifically, 
we consider the following randomized algorithm from our previous SoCG'20 paper~\cite{ChanH20}, which is a variant of a standard algorithm by Br\"onnimann and Goodrich~\cite{bronnimann1995almost} or Clarkson~\cite{clarkson1993algorithms}
(see also Agarwal and Pan~\cite{agarwal2014near}).
Below, $c_0$ is a sufficiently large constant.
The \emph{depth} of a point $p$ in a set $S$ of objects is the
number of objects of $S$ containing $p$.
A subset of objects $T\subseteq S$ is called an \emph{$\eps$-net} of $S$ if all points with depth $\geq \eps |S|$ in $S$ are covered by $T$.
The size $|\hat{S}|$ of a multiset $\hat{S}$ refers to the sum of the multiplicities of its elements $\sum_{i\in S} m_i$.

\begin{algorithm}[H]
\caption{MWU for set cover}\label{algo:2}
\begin{algorithmic}[1]
\State Guess a value $t\in [\OPT,2\,\OPT]$.
\State Define a multiset $\hat{S}$ where each object $i$ in $S$ initially has multiplicity $m_i=1$.
\Loop \Comment{call this the start of a new \emph{round}}
    \State Fix $\rho:=\frac{c_0 t\log n}{|\hat{S}|}$ and take a random sample $R$ of $\hat{S}$ with sampling probability $\rho$.
        \While {there exists a point $p\in X$ with depth in $R$ at most $\frac{c_0}{2}\log n$}
                \For {each object $i$ containing $p$} \Comment{call lines 6--8 a \emph{multiplicity-doubling step}}
                    \State Double its multiplicity $m_i$, i.e., insert $m_i$ new copies of object $i$ into $\hat{S}$. 
                    \State For each copy, independently decide to insert it into $R$ with probability $\rho$.
                \EndFor
                \If {the number of multiplicity-doubling steps in this round exceeds $t$}
                    \State Go to line~3 and start a new round.
                \EndIf
        \EndWhile
    \State Terminate and return a $\frac1{8t}$-net of $R$. 
\EndLoop
\end{algorithmic}
\end{algorithm}


In the standard version of MWU, the ``lightness'' condition in line~5 was whether the depth of $p$ in $\hat{S}$ is at most $\frac{|\hat{S}|}{2t}$.
The main difference in the above randomized version is that lightness is tested with respect to the sample $R$, which is computationally easier to work with---$R$ has size $O(t\log n)$ with high probability (w.h.p.).  Justification of this randomized variant follows from a Chernoff bound, and was shown in~\cite[Section~4.2]{ChanH20}.

It is known that the algorithm terminates in $O(\log\frac{n}{t})$ rounds and
$O(t\log \frac{n}{t})$ multiplicity-doubling steps~\cite{bronnimann1995almost,agarwal2014near,ChanH20}.
Furthermore, $|\hat{S}|$ increases by at most a factor of 2 in each round; in particular, $|\hat{S}|$ is bounded by $n^{O(1)}$ at the end of $O(\log\frac{n}{t})$ rounds.

In the standard version of MWU, line~11 returns a $\Theta(\frac 1t)$-net of the multiset $\hat{S}$, but a
$\Theta(\frac{1}{t})$-net of $R$ works just as well:
in the end, the depth in $R$ of all points of $X$ is at least $\frac{c_0}2\log n > \frac {|R|}{8t}$ w.h.p., and so $X$ will be covered by the net.
For objects that are axis-aligned squares in 2D, disks in 2D, or halfspaces in 3D,
$\eps$-nets of size $O(\frac 1\eps)$ exist, and thus the above algorithm yields a set cover of size $O(t)$, i.e., a constant-factor approximation.  By known algorithms (e.g.,~\cite{chan2016optimal}),
the net for $R$ in line~11 can be constructed in $\OO(|R|)=\OO(t)$ time.

Several modifications have been explored in previous work.
For example, in the first algorithm of Agarwal and Pan~\cite{agarwal2014near}, each round examines all points of $X$ in a fixed order and test for lightness of the points one by one (based on the observation that a point found to have large depth will still have large depth by the end of the round).
In our previous paper~\cite{ChanH20}, we have also added a step at the beginning of each round, where the multiplicities are rescaled and rounded, so as to keep
$|\hat{S}|$ bounded by $O(n)$.
These modifications led to a number of different static implementations running in $\OO(n)$ time.


\PAPER{\section{Axis-Aligned Squares}}
\SOCG{\section{Axis-aligned squares}}
Our first result for dynamic set cover is for axis-aligned squares. Previously, a sublinear time algorithm for dynamic set cover was known only for unit squares.
Our method will be divided into two cases: when OPT is small and when OPT is large. 
%
Let $t$ be a guess on OPT\@. We will 
run our algorithm for each possible $t=2^i$
in parallel. 
(When the guess is wrong, our algorithm will be able to tell whether $t$ is approximately smaller or larger than $\OPT$ w.h.p.)
The update time will  increase only by a factor of $O(\log n)$.


\PAPER{\subsection{Algorithm for Small OPT}\label{square_algo:small}}
\SOCG{\subsection{Algorithm for small OPT}\label{square_algo:small}}
Our algorithm for the small OPT case will be based on the randomized MWU algorithm described in Section~\ref{sec:mwu}.
The key is to realize that this algorithm can actually be implemented to run in \emph{sublinear} time, assuming that the points and the objects have been preprocessed in standard range searching data structures.  Since these structures are dynamizable, we can just re-run the MWU algorithm from scratch after every update.

\subsubsection{Data structures}
Our data structures are simple.
We store the point set $X$ in the standard 2D \emph{range tree}~\cite{BerBOOK}.  
For each square $s$ with center $(x,y)$ and side length $2z$, map $s$ to a point
$s^\uparrow=(x-z,x+z,y-z,y+z)$ in 4D\@.  We also store the lifted point set $S^\uparrow=\{s^\uparrow: s\in S\}$ in a 4D range tree.
Range trees support insertions and deletions in $X$ and $S$ in polylogarithmic time.

\subsubsection{Computing a solution}
We now show how to compute an $O(1)$-approximate solution in sublinear time when OPT is small
by running Algorithm~\ref{algo:2} using the above data structures.
At first glance, linear time seems unavoidable: (i)~the obvious way to find low-depth points in line~5 is to scan through all points of $X$ in each round (as was done in previous algorithms~\cite{agarwal2014near,ChanH20}), and (ii)~explicitly maintaining the multiplicities of all points would also require linear time.

To overcome these obstacles, we observe that (i)~we can use data structures to 
find the next low-depth point $p$ without testing each point one by one (recall that there are only $\OO(t)$ multiplicity-doubling steps), and (ii)~multiplicities do not need to be maintained explicitly, so long as in line~8 we can generate a multiplicity-weighted random sample among the objects containing a given point $p$ efficiently 
(recall that the sample $R$ has only $\OO(t)$ size).
The subproblem in (ii) is a \emph{weighted range sampling} problem.




\subparagraph*{Finding a low-depth point.} 
Let $b:=\frac{c_0}{2}\log n$.
Each time we want to find a low-depth point in line~5,
we compute (from scratch) ${\cal L}_{\le b}(R)$, the $(\le b)$-level of $R$, i.e., the collection of all cells in the arrangement of the squares of depth at most $b$.  It is known~\cite{Sharir91} that
${\cal L}_{\le b}(R)$ has $O(|R|b)$ cells and can be constructed in $\OO(|R|b)$ time, which is $\OO(t)$ (since $|R|=\OO(t)$
and $b=\OO(1)$).
To find a point $p$ of $X$ that has depth in $R$
at most $b$, we simply examine each cell of ${\cal L}_{\le b}(R)$, and perform an orthogonal range query to test if the cell contains a point of $X$.  All this takes $\OO(t)$ time.  

As there
are $\OO(t)$ multiplicity-doubling steps, the total cost is $\OO(t^2)$.


\subparagraph*{Weighted range sampling.}
For each square $s$ with center $(x,y)$ and side length $2z$, define its \emph{dual} point $s^*$ to be $(x,y,z)$ in 3D\@.
For each point $p=(p_x,p_y)$, define its \emph{dual} region $p^*$ to be $\{(x,y,z): z \ge \max\{|x-p_x|,|y-p_y|\}\}$ in 3D\@.  Then a point $p$ is in the square $s$ iff the point $s^*$ is in the region $p^*$.

Let $Q$ be the set of all points $p$ for which we have performed  multiplicity-doubling steps thus far.
Note that $|Q|=\OO(t)$.
Let $b'=c_0'\log n$ for a sufficiently large constant $c_0'$.
Each time we perform a multiplicity-doubling step,
we compute (from scratch) ${\cal L}_{\le b'}(Q^*)$, the $(\le b')$-level of the dual regions $Q^*=\{p^*: p\in Q\}$.  This structure corresponds to planar order-$(\le b')$ $L_\infty$ Voronoi diagrams. 
By known results~\cite{Sharir91,Chan00SICOMP,LiuPL11}, ${\cal L}_{\le b'}(Q^*)$ has $O(|Q|(b')^2)$ cells and can be constructed
in $\OO(|Q|(b')^2)$ time, which is $\OO(t)$  (since $|Q|=\OO(t)$
and $b'=\OO(1)$).
The multiplicity of a square $s\in S$ is
equal to $2^{\mbox{\scriptsize depth of $s^*$ in $Q^*$}}$ (since each $p^*$ in $Q^*$ containing $s^*$ doubles the multiplicity of $s$).  In particular, since multiplicities are bounded by $n^{O(1)}$, the depth (i.e., level) of $s^*$ must be logarithmically bounded.  So, each $s^*$ is covered by ${\cal L}_{\le b'}(Q^*)$, and the multiplicity of $s$ is determined by which cell of ${\cal L}_{\le b'}(Q^*)$ the point $s^*$ is in.  

To generate a multiplicity-weighted sample of the squares containing $p$ for line~8, after $p$ has been inserted to $Q$, we examine all cells of ${\cal L}_{\le b'}(Q^*)$ contained in $p^*$.  For each such cell $\gamma$, we identify the squares $s\in S$ for which $s^*\in\gamma$; this reduces to an orthogonal range query in $S^\uparrow$, and the answer can be expressed as a disjoint union of $\OO(1)$ canonical subsets.  Knowing the sizes and multiplicities of these canonical subsets for all such $\OO(t)$ cells, we can then generate the weighted sample in time $\OO(t)$ plus the size of the sample.

Hence, the total cost of all $\OO(t)$ multiplicity-doubling steps is $\OO(t^2)$.

In addition, we need to generate a new weighted sample $R$ (with a new sampling probability~$\rho$) in line~4 at the beginning of each round; this can be done similar to above, in $\OO(t)$ time plus the size of the sample $(\OO(t))$, for each of the $O(\log n)$ rounds.
As mentioned, the final net computation in line~11 takes $\OO(t)$ time.  We have thus obtained:


\begin{lemma}
There exists a data structure for the dynamic set cover problem for $O(n)$ axis-aligned squares and $O(n)$ points in 2D that supports insertions and deletions in $\OO(1)$ time and can find
an $O(1)$-approximate solution w.h.p.\ to
the set cover problem in $\tilde{O}(\OPT^2)$ time.
\end{lemma}

\PAPER{\subsection{Algorithm for Large OPT}\label{square_algo:large}}
\SOCG{\subsection{Algorithm for large OPT}\label{square_algo:large}}
To complement our solution for the small OPT case, we now show that the problem also gets easier when OPT is large, mainly because
we can afford a large additive error.  We describe a different, self-contained algorithm for this case (not based on modifying MWU), interestingly by using quadtrees in a novel way.

To allow for both multiplicative and additive error, we use the term
\emph{$(\alpha,\beta)$-approximation} to refer
to a solution with cost at most $\alpha\,\OPT+\beta$.

For simplicity, we assume that all coordinates are integers bounded by $U=\mathrm{poly}(n)$.  At the end, we will comment on how to remove this assumption.

In the standard \emph{quadtree}, we start with a bounding square cell and recursively divide a square cell into four square subcells. We define the size of a cell $\Gamma$ to be the number of vertices in $\Gamma$ among the squares of $S$, plus the number of points of $X$ in $\Gamma$. We stop subdividing when a leaf cell has size at most $b$, where $b$ is a parameter to be set later. This yields a subdivision into
$O(\frac{n}{b})$ cells per level, and
$O(\frac{n}{b}\log U)$ cells in total.
The quadtree decomposition can be easily made dynamic under
insertions and deletions of points and squares.


For each leaf cell $\Gamma$, a square in $S$ intersecting $\Gamma$ is called \emph{short} if at least one of its vertices is in the cell, and \emph{long} otherwise, as shown in Fig.~\ref{fig:1}. Note that a long square can have at most one side crossing the cell, because the quadtree cell $\Gamma$ is also a square. 
The union of the long squares within the cell is defined by at most 4 long squares---call these the \emph{maximal} long squares.
 (If there is a square containing $\Gamma$, we can designate one such square as the maximal long square.)
For each leaf cell $\Gamma$, it suffices to approximate the optimal set cover for the input points in $X\cap\Gamma$ using only the short squares plus the at most 4 maximal long squares in $\Gamma$. 
By charging each square in the optimal solution to the cells containing its 4 vertices, we see that the sum of the sizes of the optimal covers in the leaf cells is at most $4\,\OPT+O(\frac{n}{b}\log U)$, which is indeed an $O(1)$-approximation if
we choose $b\ge \frac{n\log n}{\rm OPT}$.


\begin{figure}[htbp]\centering
\includegraphics[scale=0.5]{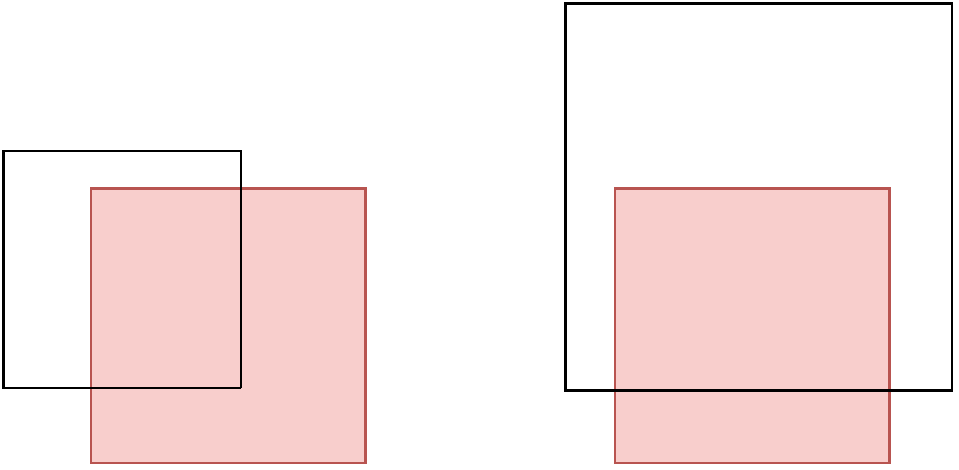}\\
\caption{Short square (left) and long square (right). The quadtree cell is shaded.}\label{fig:1}
\end{figure}

\newcommand{\DS}{{\cal S}}

Note that the complement of the union of the at most 4 maximal long squares in a cell $\Gamma$ is a rectangle $r_\Gamma$. We will store the short squares in the cell $\Gamma$ in a data structure $\DS_\Gamma$
to answer the following type of query:
\begin{quote}
    Given any query rectangle $r$, compute an $O(1)$-approximation to the optimal set cover for the points in $X\cap r$ using only the short squares.
\end{quote}
Assuming the availability of such data structures $\DS_\Gamma$, we can solve the dynamic set cover problem as follows:

\begin{itemize}
\item 
An insertion/deletion of a square $s$ in $S$ requires updating 4 of these data structures $\DS_\Gamma$ for the leaf cells $\Gamma$ containing the $4$ vertices of $s$, and also updating the maximal long squares for all leaf cells in $\tilde{O}(\frac{n}{b})$ time. 
\item
An insertion/deletion of a point $p$ in $X$ requires updating one data structure $\DS_\Gamma$ for the leaf cell $\Gamma$ containing $p$. 
\item
Whenever we want to compute a set cover solution, we examine
all $\OO(\frac{n}{b})$ cells $\Gamma$ and
query the data structure $\DS_\Gamma$ for $r_\Gamma$, and return the union of the answers.
\end{itemize}

\subparagraph*{First implementation of $\DS_\Gamma$.} A simple way to implement the data structure $\DS_\Gamma$ is as follows: in an update, we just recompute an approximate solution from scratch for every possible query rectangle in $\Gamma$. Since there are only $O(b^4)$ combinatorially different query rectangles, and static
approximate set cover on $O(b)$ squares and points takes $\OO(b)$ time~\cite{agarwal2014near,ChanH20}, the update time is $\tilde{O}(b^5)$, and the query time is trivially $O(1)$. 

As a result, the cost of  insertion/deletion in the overall method is $\tilde{O}(b^5+\frac{n}{b})$.

\subparagraph*{Improved implementation for $\DS_\Gamma$.} We further improve the update time for the data structure $\DS_\Gamma$. 
Instead of recomputing solutions for all $O(b^4)$ rectangles, the idea is to recompute solutions for a smaller number of ``canonical rectangles''.
More precisely, by using a 2D range tree~\cite{AgaEriSURV,BerBOOK} for the $O(b)$ points in $X\cap\Gamma$, with branching factor $a$ in each dimension, we can
form a set of 
canonical rectangles with total size $O(a^{O(1)}b(\log_a b)^2)$,
such that every query rectangle can be decomposed into $O((\log_a b)^2)$ canonical rectangles, ignoring portions that are empty of points.  We set $a:=b^\delta$ for an arbitrarily small constant $\delta>0$.

For each canonical rectangle $r$ with size $b_i$, there are at most $O(b_i)$ maximal long squares with respect to $r$: the union of the long squares that cut across $r$ horizontally
have at most two
edges between any two consecutive points/vertices; and a similar statement holds for the long squares that cut across $r$ vertically.
These maximal long squares can be found in $\OO(b_i)$ time by standard orthogonal range searching.
We can thus approximate the optimal set cover for the points in the canonical rectangle~$r$, using the $O(b_i)$ short squares and maximal long squares with respect to~$r$, in $\tilde{O}(b_i)$ time by known static set cover algorithms \cite{agarwal2014near,ChanH20}. 
The total time over all canonical rectangles
is  $\OO(a^{O(1)}b(\log_a b)^2)=\OO(b^{1+O(\delta)})$. 


Given a query rectangle, we can decompose it into
$O((\log_a b)^2)=O(1)$ canonical rectangles and return the union of the optimal solutions in the canonical rectangles, which is an $O(1)$-approximation (more precisely, an $O(\delta^{-2})$-approximation).

As a result, the cost of  insertion/deletion in the overall method is $\tilde{O}(b^{1+O(\delta)}+\frac{n}{b})$.  





\subparagraph{Removing the dependency on $U$.} When $U$ may be large, we can reduce the tree depth from $O(\log U)$ to $O(\log n)$ by replacing the quadtree with the \emph{BBD tree} of Arya et al.~\cite{DBLP:journals/jacm/AryaMNSW98}. Each cell in the BBD tree is the set difference of two quadtree squares, one contained in the other.  The size of each child cell is at most a fraction of the size of the parent cell.  As before, we stop subdividing when a leaf cell has size at most $b$.  Since any such leaf cell has size $\Theta(b)$ now, the number of leaf cells is $O(\frac{n}{b})$.  The BBD tree can be maintained dynamically in polylogarithmic time (for example, by periodically rebuilding when subtrees become unbalanced).
Since a leaf cell $\Gamma$ is the difference of two quadtree squares, it is not difficult to see that the number of maximal long squares in $\Gamma$ remains $O(1)$.  So, 
our previous analysis remains valid.

\begin{lemma}
Given a parameter $b$, there exists a
data structure for the dynamic set cover problem for $O(n)$ axis-aligned squares and $O(n)$ points in 2D
that maintains an $(O(1),O(\frac nb))$-approximate solution  with $\tilde{O}(b^{1+O(\delta)} + \frac{n}{b})$ insertion and deletion time.
\end{lemma}

\subparagraph*{Combining the algorithms.} When $\OPT\leq n^{1/3}$, we use the algorithm for small OPT; the running time is $\OO(\OPT^2)\le \OO(n^{2/3})$. When $\OPT >n^{1/3}$, we use the algorithm for large OPT with $b=n^{2/3}$, so that an $(O(1),O(\frac nb))$-approximation is indeed an $O(1)$-approximation;
the running time is $\OO(b^{1+O(\delta)}+\frac n{b})=O(n^{2/3+O(\delta)})$.


\begin{theorem}
There exists a 
data structure for the dynamic  set cover problem for $O(n)$ axis-aligned squares and $O(n)$ points in 2D
that maintains an $O(1)$-approximate solution w.h.p.\ with $O(n^{2/3+\delta})$ insertion and deletion time for any constant $\delta>0$.
\end{theorem}

The case of fat rectangles can be reduced to squares, since such rectangles can be replaced by $O(1)$ squares (increasing the approximation factor by only $O(1)$).
Our approach can be modified to work more generally for homothets of a fixed fat convex polygon with a constant number of vertices.






\ignore{
\subparagraph*{Approximate the size of OPT.} If we only want to approximate the size of OPT, we can probably improve the running time. The optimal set cover solution in a random quadtree cell should have expected value $O(\frac{t}{n/s})$, and worst-case $O(s)$. After normalizing $X_i$ to $[0,1]$, $\mathbb{E}[X_i]=O(\frac{t}{n})$. Using Chernoff bound, to get an $O(1)$-approximation, we need to sample $O(\frac{n}{t}\log n)$ quadtree cells.
}

\section{Halfspaces in 3D}
In this section, we study dynamic geometric set cover for the more challenging case of  3D halfspaces. Using the standard lifting transformation~\cite{BerBOOK}, we can transform 2D disks to 3D upper halfspaces.
For simplicity, we assume that all halfspaces are upper halfspaces;
Section~\ref{sec:upperlower} discusses how
to modify our algorithms when there are both upper and lower halfspaces.
Our method will be divided into three cases: small, medium, and large OPT.

\ignore{
We need the following theorems:

\begin{theorem}[partition theorem \cite{Matousek91}]\label{thm:1}
We can partition a set $P$ of $n$ points in 3D into $r$ subsets $P_1,\dots,P_r$, each of size $\Theta(\frac{n}{r})$ s.t.\ any halfspace crosses $O(r^{2/3})$ cells.
\end{theorem}
by Matou\v{s}ek. Chan showed it can be constructed in $O(n\log n)$ time \cite{Chan10}.

\begin{theorem}[shallow partition theorem \cite{matousek1992reporting}]\label{thm:2}
We can partition a set $P$ of $n$ points in 3D into $r$ subsets $P_1,\dots,P_r$, each of size $\Theta(\frac{n}{r})$ s.t.\ any halfspace that contain $O(\frac{n}{r})$ points in $P$ crosses $O(\log r)$ cells.
\end{theorem}

If $r<n^{\epsilon}$, the shallow partition can be constructed in $O(n\log n)$ time \cite{matousek1992reporting}. For larger $r$, we can recursively construct the shallow partition. The total crossing number is $(\log r)^{O(\log\log n)}=2^{O((\log\log n)^2)}$. We can also construct in $O(n^{1+\delta})$ time for any $r$ \cite{}.
}

\PAPER{\subsection{Algorithm for Small OPT}\label{algo:small}}
\SOCG{\subsection{Algorithm for small OPT}\label{algo:small}}

Similar to the small OPT algorithm for axis-aligned squares in Section~\ref{square_algo:small},
we describe a small OPT algorithm for halfspaces based on the randomized MWU algorithm in Section~\ref{sec:mwu}.
Although our earlier approach using levels in arrangements could be generalized, we describe a better approach based on augmenting partition trees with counters.

\subsubsection{Data structures}

\newcommand{\D}{\Gamma}
\newcommand{\PAR}{\mathop{\rm parent}}
We store the 3D point set $X$
in Matou\v sek's \emph{partition tree}~\cite{matouvsek1992efficient}: The tree has height $O(\log n)$ and degree $r$ for a sufficiently large constant $r$.  Each node $v$
stores a simplicial cell $\D_v$ and a ``canonical subset'' $X_v\subset  \D_v$, where $X_v=X$ at the root $v$, and
$\D_v$ is contained in $\D_{\PAR(v)}$,
$X_v$ is the disjoint union of $X_{v'}$ over all children $v'$ or $v$, and $X_v$ has constant size at each leaf~$v$.
Furthermore, any halfspace crosses $O(n^{2/3+\delta})$ cells of the tree for any arbitrarily small constant $\delta>0$ (depending on $r$).
Here a halfspace $h$ \emph{crosses} a cell $\D$ iff the boundary of $h$ intersects $\D$.

For each upper halfspace $h$, let $h^*$ denote its dual point;
for each point $p$, let $p^*$ denote its dual upper halfspace.
(Duality~\cite{BerBOOK} is defined so that $p$ is in $h$ iff $h^*$ is in $p^*$.)

We also store the 3D dual point set $S^*=\{h^*:h\in S\}$ in Matou\v sek's partition tree.
Each node $v$ stores a cell $\Gamma_v$ and a canonical subset $S^*_v\subset \Gamma_v$ like above.

Matou\v sek's partition trees can be built in $\OO(n)$ time and support insertions and deletions in $X$ and $S$ in polylogarithmic time.
(In the static case, there are slightly improved partition trees reducing the $n^\delta$ factor in the crossing number bound \cite{matouvsek1992efficient,matousek93DCG,Chan12DCG}, but these will not be important to us.)

\subsubsection{Computing a solution}
We now show how to compute an $O(1)$-approximate solution in sublinear time when OPT is small
by running Algorithm~\ref{algo:2} using the above data structures.
As in Section~\ref{square_algo:small}, the main subproblems are (i) finding a low-depth point with respect to $R$, and (ii) weighted range sampling, where the weights are the multiplicities (which are not explicitly stored).

\subparagraph*{Finding a low-depth point.}
We maintain two values at each node $v$ of the partition tree of~$X$:

\begin{itemize}
    \item $c_v$ is the number of halfspaces of $R$ containing $\Gamma_v$ but not containing $\Gamma_{\PAR(v)}$.
 
    \item $d_v$ is the minimum depth among all points in $X_v$ with respect to the halfspaces of $R$ crossing $\Gamma_v$.
\end{itemize}

\noindent
The overall minimum depth with respect to $R$ is given by the value $d_v$ at the root $v$.  Whenever we insert a halfspace $h$ to $R$, for each of the $O(n^{2/3+\delta})$ cells $\D_v$ crossed by $h$, we update the counters $c_{v'}$ for the children $v'$ of the nodes $v$; we also update the value $d_v$ bottom-up according to the formula $d_v = \min_{\mbox{\scriptsize\rm  child $v'$ of $v$}}(d_{v'}+c_{v'})$.  Thus, all values can be maintained in $O(n^{2/3+\delta})$ time per insertion to $R$.

As there are $\OO(t)$ insertions to $R$, the total cost is $\OO(t n^{2/3+\delta})$.  This cost covers the resetting of counters at every round.

(We remark that the idea of augmenting nodes of partition trees with counters appeared before in at least one prior work on dynamic geometric data structures~\cite[Theorem~4.1]{Chan03sicomp}.)

\subparagraph*{Weighted range sampling.}
Let $Q$ be the set of all points $p$ for which we have performed  multiplicity-doubling steps thus far.  Note that $|Q|=\OO(t)$.
The multiplicity of a halfspace $h\in S$ is 
$2^{\mbox{\scriptsize depth of $h^*$ in $Q^*$}}$.
To implicitly represent the multiplicities and their sum, we maintain two values at each node $v$ of the partition tree for $S^*$:

\begin{itemize}
    \item $c_v$ is the number of dual halfspaces of $Q^*$ containing $\Gamma_v$ but not containing $\Gamma_{\PAR(v)}$.
 
    \item $m_v$ is the sum of $2^{\mbox{\scriptsize depth of $h^*$ among the halfspaces of $Q^*$ crossing $\Gamma_v$}}$ over all $h^*\in S^*_v$.
\end{itemize}

\noindent
Whenever we insert a point $p$ to $Q$, 
for each of the $O(n^{2/3+\delta})$ cells $\D_v$ crossed by the dual halfspace $p^*$, we update the counters $c_{v'}$ for the children $v'$ of the nodes $v$; we also update the value $m_v$ bottom-up according to the formula $m_v = \sum_{\mbox{\scriptsize\rm  child $v'$ of $v$}} 2^{c_{v'}}m_{v'}$.  Thus, all values can be maintained in $O(n^{2/3+\delta})$ time per insertion to $Q$.

To generate a weighted sample of the halfspaces of $S$ containing $p$ for line~8, we find all $O(n^{2/3+\delta})$ cells $\D_v$ crossed by the dual halfspace $p^*$, and consider the canonical subsets $S_{v'}^*$ for the children $v'$ of $v$ with $\D_{v'}$ contained in $p^*$. We can then sample from these canonical subsets, weighted by $m_{v'} 2^{\sum_u c_u}$, where the sum is over all ancestors $u$ of $v'$.
All this takes time $O(n^{2/3+\delta})$ plus the size of the sample.

Hence, the total cost of all $\OO(t)$ multiplicity-doubling steps is $\OO(tn^{2/3+\delta})$.

\begin{lemma}
There exists a data structure for the dynamic set cover problem for $O(n)$ upper halfspaces and $O(n)$ points in 3D that supports insertions and deletions in $\OO(1)$ time and can find
an $O(1)$-approximate solution w.h.p.\ in $\OO(\OPT\cdot n^{2/3+\delta})$ time.
\end{lemma}



\ignore{
rounding the multiplicity? We can add $1$ to each $m_i$ instead. Then each $m_i$ is a double, and we sample proportional to $m_i$. This ensures that even if we don't remove an $\epsilon$-net in the beginning, each $\epsilon$-light point is also $\epsilon$-light in the original set of halfspaces $S$ (because each $m_i\geq 1$). Also we know $|\hat{S}|$, because we can maintain this after each multiplicity doubling step, so we can normalize the $m_i$'s.

We can also maintain an $\epsilon$-net in the beginning, and delete all points covered by it. Whenever we delete a halfspace from the $\epsilon$-net, we need to rebuild the $\epsilon$-net in $\tilde{O}(n)$ time. But the $\epsilon$-net we constructed is a subset of a random sample $R$ with size $O(OPT\log n)$, so this happens with probability $\tilde{O}(\frac{t}{n})$. The expected running time is $\tilde{O}(t)$. However this can only solve the oblivious adversary case. (use our new ideas.)

combine partition with cutting? but how to maintain counters using cutting?

alternative solution? look at the currently selected points that performed a multiplicity doubling step?

}

\PAPER{\subsection{Algorithm for Medium OPT}
\label{app:medium}
}
\SOCG{\subsection{Algorithm for medium OPT}
}


The preceding algorithm works well only when $\OPT$ is smaller than about $n^{1/3}$.  We show that a more involved algorithm, also based on MWU, can achieve sublinear time even when $\OPT$ approaches $n^{1-\delta}$.  The basic approach is to use \emph{shallow} versions of the partition trees~\cite{matousek1992reporting}.


\newcommand{\APPENDIXMEDIUMOPT}{

\subsubsection{Data structures}

We begin with a lemma that was
used before in some of the previous static algorithms by Agarwal and Pan~\cite{agarwal2014near} and ours~\cite{ChanH20}.
With this lemma, we can effectively
make every input halfspace shallow, i.e., contain at most $\OO(\frac nt)$ points.
The extra condition in the second sentence of the lemma below is new, and is needed in order to make dynamization possible later (difficulty arises when halfspaces in $T_0$ get deleted).  Because of this extra condition, we describe a construction which is different from the previous algorithms~\cite{agarwal2014near,ChanH20}.

\begin{lemma}\label{lem:t0}
Given a set $X_0$ of $n$ points, a set $S_0$ of $n$ halfspaces in 3D and a parameter $t$, we can construct a subset of halfspaces $T_0\subseteq S_0$ of size $O(t)$, and a subset of points
$A_0\subseteq X_0$, such that (i)~each point in $A_0$ is covered by $T_0$, and
(ii) each halfspace of $S_0-T_0$ contains $\OO(\frac nt)$ points of $X_0-A_0$.

Furthermore,  we can decompose $A_0=\bigcup_{h\in T_0}A_h$,
where $A_h$ has size $O(\frac nt)$ for each $h\in T_0$, 
such that each point in $A_h$ is covered by $h$.  
The construction takes $\OO(n)$ time.
\end{lemma}
\begin{proof}
We use a simple ``greedy'' approach:
We examine each halfspace $h\in S$ in an arbitrary order,
and test whether $h$ contains more than $\frac nt$ points of $X_0$.
If so, we add $h$ to $T_0$, pick some $\frac nt$ (but not more) points in $X_0\cap h$ to add to $A_h$, and delete $A_h$ from $X_0$.  Clearly, the number of halfspaces added to $T_0$ is $O(t)$.

To bound the construction time, we maintain $X_0$ in Chan's  dynamic 3D halfspace range reporting structure with polylogarithmic amortized update time~\cite{Chan10JACM}.  The cost of all deletions is  $\OO(n)$.  Testing each halfspace $h$ requires a halfspace
range counting query, but for the above purposes, it suffices to use an approximate count, with $\OO(1)$ approximation factor.  Given a query halfspace $h$ containing $k$ points, in $\OO(1)$ time, one can find $O(\log n)$ lists of size $O(k)$ from Chan's data structure~\cite{Chan10JACM} (see also~\cite[Section~3]{Chan12IJCGA}), so that the $k$ points inside $h$ are contained in the union of these lists.  The sum of the sizes of these lists is $O(k\log n)$ and $\Omega(k)$, yielding an $O(\log n)$-approximation of the count.
\end{proof}

We divide the update sequence into \emph{phases} of $g$ updates each, for a parameter $g\ll t$.  Our data structure will be rebuilt periodically, after each phase.  Let $X_0$ and $S_0$ be $X$ and $S$ at the beginning of the current phase.  Let $X_I$ and $S_I$ be the current set of points and the set of halfspaces that have been inserted to $X$ and $S$ in the current phase.
Let $X_D$ and $S_D$ be the current set of points and the set of halfspaces deleted from $X$ and $S$ in the current phase.
At the start of each phase, we apply Lemma~\ref{lem:t0}.

We store the 3D point set $X_0-A_0$
in a partition tree, like before.
However, we will need a bound on the
crossing number that is sensitive to
shallowness: namely, 

\begin{itemize}
\item any halfspace containing $\OO(\frac nt)$ points of $X_0-A_0$  crosses at most $O((\frac nt)^{2/3+O(\delta)})$ cells of the tree;
\item any other halfspace crosses at most $O(t\cdot  (\frac nt)^{2/3+\delta})$ cells.
\end{itemize}

\noindent
This follows by  combining  Matou\v sek's \emph{shallow} version of the partition tree~\cite{matousek1992reporting} with the original version of the partition tree:
Using the shallow version of the partition tree, with $O(t)$ leaf cells containing $O(\frac nt)$ points each, 
any halfspace that contains $\OO(\frac nt)$ points is known to cross at most $O(n^\delta)$ leaf cells~\cite{matousek1992reporting}.  For each leaf cell with $O(\frac nt)$ points, we build the original partition tree~\cite{matouvsek1992efficient}, which has crossing number bound $O((\frac nt)^{2/3+\delta})$.

In addition, we maintain the following counter at each node $v$ of the partition tree:

\begin{itemize}
    \item $c^\#_v$ is the number of halfspaces of $T_0\cup S_I - S_D$ containing $\Gamma_v$ but not containing $\Gamma_{\PAR(v)}$.
\end{itemize}

\noindent
In the beginning of a phase, we can compute each count $c^\#_v$ by $O(1)$ 3D simplex range counting query on the dual points of $T_0^*$ (since halfspaces containing a simplex and not containing another simplex dualize to points in a polyhedral region of constant size); by known results~\cite{AgaEriSURV,matouvsek1992efficient}, $O(n)$ such queries on $O(t)$ points in 3D take $\OO(n + (nt)^{3/4})$ time.  The amortized cost is $\OO(\frac{n + (nt)^{3/4}}{g})$.

Afterwards, during each insertion/deletion of a halfspace in $S$, we increment/decrement the counters at $O(t\cdot (\frac nt)^{2/3+\delta})$ nodes of the tree (note that the halfspace may not be shallow).  Thus, each update in $S$ costs $O(t^{1/3}n^{2/3+\delta})$ time.

We also store the 3D dual point set $(S_0-T_0-S_D)^*$ in a partition tree, again with $O((\frac nt)^{2/3+O(\delta)})$ crossing number bound with respect to halfspaces containing $\OO(\frac nt)$ points.  Such a partition tree supports deletions in polylogarithmic time.


\subsubsection{Computing a solution}

\newcommand{\Rextra}{R_{\mbox{\scriptsize\rm extra}}}

We now show how to compute an $O(1)$-approximate solution in sublinear time when OPT is sublinear  using the above data structures.
We first take a new unweighted random sample $\Rextra\subset S_0-T_0-S_D$ of size $t$.
We include $\Rextra\cup T_0\cup S_I-S_D$ in the solution.
Since this set has size $O(t+g)=O(t)$, this increases the approximation factor only by $O(1)$.  
Let $E=\bigcup_{h\in T_0\cap S_D} A_h \cup X_I - X_D$, which has $O(\frac{gn}{t})$ points.
It remains to cover the points in $(X_0-A_0-X_D)\cup E$, excluding those already covered by $\Rextra\cup T_0\cup S_I-S_D$, using
halfspaces in $S_0-T_0-S_D$.  We will do so by 
running Algorithm~\ref{algo:2} on these points and halfspaces using the above data structures.
As before, the main subproblems are (i) finding a low-depth point, and (ii) weighted range sampling.

\subparagraph*{Finding a low-depth point.}
We can find the minimum depth of the points of $X_0-A_0-X_D$ (with respect to $R$) as in the small OPT algorithm, by using counters in the partition tree for $X_0-A_0$.  

Since any halfspace in $S_0-T_0-S_D$ contains at most $\OO(\frac n t)$ points of $X_0-A_0-X_D$
by Lemma~\ref{lem:t0}, the cost per insertion to $R$ is reduced from $O(n^{2/3+\delta})$ to $O((\frac nt)^{2/3+O(\delta)})$.
The total cost over all $\OO(t)$ insertions to $R$ is
$O(t (\frac nt)^{2/3+O(\delta)})=O(t^{1/3}n^{2/3+O(\delta)}).$

One technicality is that we should be excluding points already covered by $\Rextra\cup T_0\cup S_I-S_D$.  To fix this, we add $c_0\log n$ copies of $\Rextra\cup T_0\cup S_I-S_D$ to $R$ at the beginning of each round.  This way, points covered by $\Rextra\cup T_0\cup S_I-S_D$ would not be picked as low-depth points.
Adding these copies of $T_0\cup S_I-S_D$ requires no extra effort, since we can initialize $c_v$ with the already computed $c_v^\#$ value, times $c_0\log n$, for all nodes~$v$ encountered.
The copies of the $\OO(t)$ halfspaces of $\Rextra$ can be inserted one by one, in $\OO((\frac nt)^{2/3+O(\delta)})$ time each,
as above.
The extra cost for these $\OO(t)$ insertions to $R$ is bounded as above.

Another technicality that we should be excluding the deleted points in $X_D$ when defining the minimum depth values $d_v$.  When we delete a point from $X$, we can update the $d_v$ values along a path bottom-up in $\OO(1)$ time.

We can find low-depth points of $E$ (with respect to $R$) more naively, by examining the points of $E$, excluding those covered by $\Rextra\cup T_0\cup S_I-S_D$, and testing them one by one in each round (like in Agarwal and Pan's algorithm~\cite{agarwal2014near}).  In line~5, it suffices to use an $O(1)$-approximation to the depth (after adjusting constants in the pseudocode), and as noted in our previous paper~\cite{ChanH20}, we can apply known data structures for 3D halfspace approximate range counting~\cite{AfshaniC09} for the dual points $R^*$; queries and insertions to $R$ take polylogarithmic time.
A point that has been found to have depth larger than the threshold will remain having large depth during a round.  The total cost over all rounds is $\OO(|E|)=\OO(\frac{gn}t)$.

\subparagraph*{Weighted range sampling.}
During a multiplicity-doubling step for a point $p$, we can
generate a multiplicity-weighted sample from the halfspaces containing $p$ in the same way as in the small OPT algorithm, by using counters in the partition tree for $(S_0-T_0-S_D)^*$.  Observe that $p$ has depth in $S_0-T_0-S_D$ at most $O(\frac nt\log n)$ w.h.p., because otherwise, $p$ would be covered by the random sample $\Rextra$ and would have been excluded (in other words, a random sample $\Rextra$ of size $t$
is a $\Theta(\frac {\log n}{t})$-net w.h.p.).
Since the dual halfspace $p^*$ contains at most $\OO(\frac nt)$ points of $S^*$, the cost per multiplicity-doubling step is reduced from $\OO(n^{2/3+\delta})$ to $\OO((\frac nt)^{2/3+O(\delta)})$.
The total cost over all $\OO(t)$ multiplicity-doubling steps is
$\OO(t^{1/3}n^{2/3+O(\delta)}).$

In conclusion, the total time for running MWU is
$\OO(\frac{gn}t + t^{1/3}n^{2/3+O(\delta)})$.  To balance this computation cost with the $\OO(\frac{n+(nt)^{3/4}}g)$ update cost,
we set $g=\frac{t^{7/8}}{n^{1/8}} + \sqrt{t}$
and get a bound of $\OO(\frac{n^{7/8}}{t^{1/8}} +\frac{n}{\sqrt{t}} + t^{1/3}n^{2/3+O(\delta)})$.

}

\APPENDIXMEDIUMOPT

\begin{lemma}\label{lem:medium}
There exists a data structure for the dynamic set cover problem for $O(n)$ upper halfspaces and $O(n)$ points in 3D that maintains
an $O(1)$-approximate solution w.h.p.\  with $\tilde{O}(
\frac{n^{7/8}}{\tiny\OPT^{1/8}} + \frac{n}{\sqrt{\tiny\OPT}}+\OPT^{1/3}n^{2/3+O(\delta)})$ amortized insertion and deletion time.
\end{lemma}


\PAPER{\subsection{Algorithm for Large OPT}\label{algo:large}}
\SOCG{\subsection{Algorithm for large OPT}\label{algo:large}}
Lastly, we give an algorithm for the large OPT case, which is very different from the algorithms in the previous subsections (and not based on modifying MWU\@).
Here,
we can only compute the size of the approximate set cover, not the cover itself.
Like before, we will show that the problem gets easier for large OPT, because we can afford a large additive error. The idea 
is to decompose the problem into subproblems via geometric sampling and planar separators, and then approximate the sum of the subproblems' answers by sampling again.

\subsubsection{Data structures}

We just store the dual point set $S^*$ in a known 3D halfspace range reporting structure.  The data structure by Chan~\cite{Chan10JACM} supports queries in $O((\log n+k)\log n)$ time for output size $k$, and 
insertions and deletions in $S$ in polylogarithmic amortized time.

We store the $xy$-projection of the point set $X$ in a known 2D triangle range searching structure~\cite{matouvsek1992efficient} that supports queries in $O(\frac{n^{1/2+\delta}}{z^{1/2}}+k)$ time for output size~$k$, and
insertions and deletions in $X$ in $\OO(z)$ time for a given trade-off parameter $z\in [1,n]$.

\subsubsection{Approximating the optimal value}

\newcommand{\VD}{\textrm{VD}}

Let $b$ and $g$ be parameters to be set later. Take a random sample $R$ of the halfspaces $S$ with size $\frac{n}{b}$.
Imagine that $R$ is included in the solution. The remaining uncovered space is the complement of the union of $R$, which is a 3D convex polyhedron. There are $O(|R|)=O(\frac{n}{b})$ cells in the vertical decomposition $\VD(R)$ of this polyhedron
(formed by triangulating each face and drawing a vertical wall at each edge of the triangulation).
Each cell is crossed by $O(b\log n)$ halfspaces w.h.p., 
by well-known geometric sampling analysis~\cite{Clarkson87}.  The decomposition $\VD(R)$ can be constructed in $\OO(\frac nb)$ time.



Our key idea is to use planar graph separators to divide into smaller subproblems.
The following is a multi-cluster version of the standard planar separator theorem~\cite{LiptonT80} (sometimes known as ``$r$-divisions''~\cite{federickson1987fast}):

\begin{lemma}[Planar Separator Theorem, Multi-Cluster Version]\label{thm:3}
Given a planar graph $G=(V,E)$ with $n$ vertices, and a parameter $g$, we can partition $V$ into $\frac ng$ subsets $V_1,\cdots,V_{n/g}$ of size $O(g)$ each, and an extra ``boundary set'' $B$ of size $O(\frac n{\sqrt{g}})$, such that no two vertices from different subsets $V_i$ and $V_j$ are adjacent.
The partition can be constructed in $\OO(n)$ time.
\end{lemma}

(We remark that the general idea of combining cuttings/geometric sampling with planar graph separators  appeared in some geometric approximation algorithms before, e.g., \cite{AdamaszekHW19}.)

We apply Lemma~\ref{thm:3} to the dual graph of $\VD(R)$ (which has size $O(\frac{n}{b})$), yielding $O((n/b)/g)$ ``clusters'' of $O(g)$ cells each, and a set $B$ of $O((n/b)/\sqrt{g})$ ``boundary cells'',
in $\OO(\frac nb)$ time.

Let $S_B$ be the subset of all halfspaces of $S$ that cross boundary cells of $B$.  Note that
$|S_B|=O((n/b)/\sqrt{g}\cdot b\log n)=\OO(\frac{n}{\sqrt{g}})$ w.h.p.

For each cluster $\gamma$, let $X_\gamma$ denote the subset of all points of $X$ whose $xy$-projections lie in the $xy$-projection of the cells of $\gamma$, and let $S_\gamma$ denote the subset of all halfspaces of $S$ that cross the cells of $\gamma$.
Note that $|S_\gamma|=O(g\cdot b\log n)=\OO(bg)$ w.h.p.
Let $\OPT_\gamma$ denote the optimal value for the set cover problem for the halfspaces of $S_\gamma$ and the points of $X_\gamma$ not covered by
$R\cup S_B$.

\begin{claim}
$\sum_\gamma \OPT_\gamma$ approximates $\OPT$ with
additive error $\OO(\frac{n}{b}+\frac{n}{\sqrt{g}})$ w.h.p.
\end{claim}
\begin{proof}
A feasible solution can be formed by
taking the union of the solutions corresponding to $\OPT_\gamma$, together with $R$ (to cover points not covered by $\VD(R)$) and $S_B$ (to cover points inside boundary cells).  As $|R|=\frac nb$ and $|S_B|=\OO(\frac{n}{\sqrt{g}})$ w.h.p.,
this proves that $\OPT\le \sum_\gamma\OPT_\gamma + \OO(\frac{n}{b}+\frac{n}{\sqrt{g}})$.

In the other direction,
observe that if a halfspace $h$ crosses two different clusters $\gamma_i$ and $\gamma_j$, it must also cross some boundary cell in $X$
by convexity: pick points $p\in h\cap\gamma_i$ and $q\in h\cap\gamma_j$; then the line segment $\overline{pq}$ must hit the wall of some boundary cell.
So, after removing  $R\cup S_B$ from the global optimal solution, we get
disjoint local solutions in the clusters.  This proves that
$\OPT\ge
\sum_\gamma\OPT_\gamma$.
\end{proof}



We use the following known fact about approximating 
a sum via random sampling (which is of course a standard trick):

\begin{lemma}\label{lem:chernoff}
Suppose $a_1+\cdots + a_m=T$ where $a_i\in [0,U]$.
Take a random subset $R$ of $r=\frac{(c_0/\eps^2)mU\log n}{T}$ elements from $a_1,\ldots,a_m$, for a sufficiently large constant $c_0$.
Then $\sum_{a_i\in R}a_i\cdot \frac rm$ is a $(1+\eps)$-approximation to $T$ w.h.p.
\end{lemma}
\begin{proof}
By rescaling the $a_i$'s and $T$ by a factor $U$, we may assume that $U=1$.
Define a random variable $Y_i$, which is $a_i$ with probability $\frac rm$, and 0 otherwise.
Then $E[\sum_i Y_i]=T\cdot \frac rm = (c_0/\eps^2)\log n$.  The result follows from a standard Chernoff bound on the $Y_i$'s.
\end{proof}

By applying the above lemma with $m=(n/b)/g$, $T=\Theta(t)$, and $U=\OO(bg)$ (assuming $\OPT$ is finite), we can $O(1)$-approximate $\OPT$
by summing $\OPT_\gamma$
over a random sample of
$r = \OO(\frac{mU}{T})=\OO(\frac nt)$
clusters $\gamma$.

We can generate the set $S_B$ by
finding the halfspaces of $S$ that contain the $O((n/b)/\sqrt{g})$ vertices of the cells in $B$---this corresponds to $O((n/b)/\sqrt{g})$ halfspace range reporting queries for the dual 3D point set $S^*$, each with output size $\OO(b)$ w.h.p.\ and each taking $\OO(b)$ time~\cite{Chan10JACM}.
Thus, $S_B$ can be found in $\OO(\frac n{\sqrt{g}})$ time.
We compute the union of $R\cup S_B$, which is the complement of an intersection of halfspaces, by the dual of 3D convex hull algorithm~\cite{BerBOOK}.
This takes $\OO(\frac nb + \frac n{\sqrt{g}})$ time.

For each chosen cluster $\gamma$,
we can generate $S_\gamma$ similarly by $O(g)$ halfspace range reporting queries for $S^*$, each with output size $\OO(b)$ w.h.p.
Thus, $S_\gamma$ can be found in $\OO(bg)$ time.
We can generate $X_\gamma$ by 
performing $O(g)$ triangle range reporting queries for the 2D $xy$-projection of the point set $X$.
Thus, $X_\gamma$ can be found in $\OO(g\frac{n^{1/2+\delta}}{z^{1/2}}+ |X_\gamma|)$ time.
We filter points of $X_\gamma$ covered by $R\cup S_B$, by performing $|X_\gamma|$ planar point location queries~\cite{BerBOOK} in the $xy$-projection
of the boundary of the union of $R\cup S_B$.  This takes
$\OO(|X_\gamma|)$ time.
We can then compute an $O(1)$-approximation to $\OPT_\gamma$ by running a known static set cover algorithm~\cite{agarwal2014near,ChanH20} in $\OO(bg+|X_\gamma|)$ time.

The expected sum of $X_\gamma$ over all chosen clusters $\gamma$ is $O(n\cdot \frac rm)=O(rbg)$.
The total expected time over $r$ clusters is
$\OO(rbg + rg\frac{n^{1/2+\delta}}{z^{1/2}})
=\OO(\frac{bgn}t + \frac{gn^{3/2+\delta}}{tz^{1/2}})$.
The overall expected running time
is $\OO(\frac n b + \frac{n}{\sqrt{g}} + \frac{bgn}t + \frac{gn^{3/2+\delta}}{tz^{1/2}})$, and we obtain an
$(O(1),\OO(\frac nb + \frac n{\sqrt{g}}))$-approximation.
(The expected bound can be converted to worst-case by placing a time limit and re-running logarithmically many times.)
Choosing $g=b^2$ yields the following result:



\begin{lemma}
Given parameters $b$ and $z\in [1,n]$ and any constant $\eps>0$, there exists a data structure for the dynamic set cover problem for
$O(n)$ upper halfspaces and $O(n)$ points in 3D that supports insertions and deletions in $\OO(z)$ amortized time and can find the value of an $(O(1),\OO(\frac nb))$-approximation w.h.p.\ in 
$\OO(\frac nb + \frac{b^3n}{\tiny \OPT} + \frac{b^2n^{3/2+\delta}}{{\tiny \OPT}\cdot z^{1/2}})$ time for any constant $\delta>0$.
\end{lemma}

A minor technicality is that when applying Lemma~\ref{lem:chernoff},
 we have assumed that the optimal value is finite. 
The problem of checking whether a solution exists, i.e.,
whether a point set is covered by a set of halfspaces (or more generally, maintaining the lowest-depth point), subject to insertions and deletions of points and halfspaces, has already been solved before by Chan \cite[Theorem~4.1]{Chan03sicomp},
who gave a fully dynamic algorithm with $\tilde{O}(n^{2/3})$ time per operation (based on augmenting partition trees with counters, similar to what we have done here).

\subparagraph*{Combining the algorithms.} Finally, we combine all three algorithms:

\begin{enumerate}
\item 
When $\OPT\le n^{2/9}$, we use the algorithm for small OPT;
the running time is 
$\OO(\OPT\cdot n^{2/3+\delta})\le \OO(n^{8/9+\delta})$.
\item 
When $n^{2/9}<\OPT\leq n^{10/13}$, we use the algorithm for medium OPT;
the running time is
$\OO(
\frac{n^{7/8}}{\tiny\OPT^{1/8}} + \frac{n}{\sqrt{\tiny\OPT}}+\OPT^{1/3}n^{2/3+O(\delta)})
\le O(n^{12/13+O(\delta)}).$
\item
When $\OPT > n^{10/13}$, we use the algorithm for large OPT with $b=\tilde{\Theta}(n^{3/13})$ and $z=n^{7/13}$, so that an $(O(1),\OO(\frac nb))$-approximation is indeed an $O(1)$-approximation; the running time is $\OO(\frac nb + \frac{b^3n}{\tiny\OPT} + \frac{b^2n^{3/2+\delta}}{{\tiny\OPT}\cdot z^{1/2}})\le O(n^{12/13+O(\delta)})$.
\end{enumerate}

\ignore{
When  When $t>n^{\frac{10}{13}}$, we use the algorithm f The final running time is $\tilde{O}(\max\{n^{5/6},n^{12/13}\})=\tilde{O}(n^{12/13})$.

Remark. We can bootstrap. Suppose the final running time is $O(n^{\alpha})$. (since $b<bg$, we will rebuild the clusters early.) When we recurse in a cluster, we kind of know which case (small/medium/large OPT) we are in, because we can terminate when the total size of the optimal solutions in the sample reaches $\tilde{\Omega}(gb)$. (do we still need to maintain the data structures in the other cases?)

We can get a truly sublinear time algorithm for all ranges of OPT. (?)
}

\begin{theorem}\label{thm:halfspaces}
There exists a 
data structure for the dynamic set cover problem for $O(n)$ upper halfspaces and $O(n)$ points in 3D
that maintains the value of an $O(1)$-approximate solution w.h.p.\  with $O(n^{12/13+\delta})$ amortized insertion and deletion time for any constant $\delta>0$.
\end{theorem}

\subsection{Upper and Lower Halfspaces}\label{sec:upperlower}

\subparagraph*{Small and medium OPT.}
It is straightforward to modify the small and medium OPT
algorithms in Sections~\ref{algo:small} and \ref{app:medium}
to handle the case when there are both upper and lower halfspaces in $S$.
For weighted range sampling, we can handle the upper halfspaces and the lower halfspaces separately; for example, we build the partition tree for the dual points of $S$ separately for the upper and the lower halfspaces.
To find low-depth points, we use just one partition tree for the points of $X$, where depth is defined relatively to the combined set of upper and lower halfspaces.

\subparagraph*{Large OPT.}
In the large OPT algorithm, we can no longer use the vertical decomposition.  Instead, we pick a point $p_0$ inside the polyhedron and consider a ``star'' triangulation of the polyhedron where all tetrahedra have $p_0$ as a vertex.  Because we can no longer use $xy$-projections, naively we would need to replace 2D triangle range searching with 3D simplex range searching, which would increase the update time slightly.  

If $p_0$ is fixed, we can replace the orthogonal $xy$-projection with a perspective projection with respect to $p_0$, and we can still use 2D triangle range searching.  We describe a way to find a point $p_0$ that stays fixed for a number of updates.  (Note that $p_0$ need not be in $X$.)

Specifically, we divide the update sequence into phases with $\frac nb$ updates each.
At the beginning of each phase, we set $p_0$ to be a point of minimum depth with respect to $S$, among all points in $\mathbb{R}^3$; an $O(1)$-approximation is fine and can be found in $\OO(n)$ randomized time~\cite{AfshaniC09,AronovH08}.

If $p_0$ has depth at least $\frac {2n}b$ at the beginning of the phase, the minimum depth  is at least $\frac nb$ during the entire phase, and
a $(\frac 1b)$-net of size $\OO(b)$ is a set cover and can be generated by random sampling; trivially, this gives an $(O(1),\OO(\frac nb))$-approximation, assuming $b\le \sqrt{n}$.  

Otherwise, $p_0$ has depth $O(\frac nb)$ during the entire phase.
In the large OPT algorithm, we let $Z_0$ be the set of $O(\frac nb)$ halfspaces containing $p_0$ (which can be found by halfspace range reporting in the dual), include $Z_0$ in the solution, and remove $Z_0$ from $S$ before taking the random sample $R$.  As a result, the complement of the union of $R$ indeed contains $p_0$.  The rest of the algorithm is similar, using a perspective projection from $p_0$ instead of $xy$-projection.  (Points covered by $Z_0$ should be excluded, and we can do so by adding $Z_0$ to $S_B$.)
The additive error increases by $O(\frac nb)$, and so is asymptotically unchanged.

The data structures have preprocessing time $\OO(nz)$.  Since we rebuild after every $\frac nb$ updates, the amortized update cost is $\OO(bz)$.
In our application with $b=\tilde{\Theta}(n^{3/13})$ and $z=n^{7/13}$, this cost does not dominate.

\ignore{

********

In the large OPT algorithm, we replace the vertical decomposition with another standard decomposition, the bottom-vertex triangulation~\cite{Clarkson87}.
The computation of the sets $X_\gamma$ now requires 3D simplex range searching (we can't work with the $xy$-projection in 2D).  The query time for 3D simplex range searching~\cite{matouvsek1992efficient} is
$O(\frac{n^{2/3+\delta}}{z^{1/3}})$, for trade-off parameter $z\in [1,n^2]$.
As a result, the time bound increases 
to $\OO(\frac nb + \frac{b^3n}{\tiny \OPT} + \frac{b^2n^{5/3+\delta}}{{\tiny \OPT}\cdot z^{1/3}})$.

\subparagraph*{Combining the algorithms.}
When $n^{2/9}<\OPT\le n^{25/31}$, we use the algorithm for medium OPT; the running time is 
the running time is
$\OO(
\frac{n^{7/8}}{\tiny\OPT^{1/8}} + \frac{n}{\sqrt{\tiny\OPT}}+\OPT^{1/3}n^{2/3+O(\delta)})
\le O(n^{29/31+O(\delta)}).$
When $\OPT > n^{25/31}$, we use the algorithm for large OPT with $z=n^{29/31}$ and $b=\tilde{\Theta}(n^{6/31})$, so that an $(O(1),\OO(\frac nb))$-approximation is indeed an $O(1)$-approximation; the running time is $\OO(\frac nb + \frac{b^3n}{\tiny\OPT} + \frac{b^2n^{5/3+\delta}}{{\tiny\OPT}\cdot z^{1/3}})\le O(n^{29/31+O(\delta)})$.

\begin{theorem}
There exists a 
data structure for the dynamic set cover problem for $O(n)$ upper and lower halfspaces and $O(n)$ points in 3D
that maintains the value of an $O(1)$-approximate solution w.h.p.\  with $O(n^{29/31+\delta})$ amortized insertion and deletion time for any constant $\delta>0$.
\end{theorem}

}

\begin{theorem}
Theorem~\ref{thm:halfspaces} holds even when there are both upper and lower halfspaces.
\end{theorem}


\section{Improving Static Set Cover}\label{sec:static}

In this last section, we show how the techniques we have developed for
the dynamic geometric set cover problem can lead to a randomized algorithm for static set cover for 3D halfspaces running in $O(n\log n)$ time, which is optimal and improves our previous $O(n\log n(\log\log n)^{O(1)})$ randomized algorithm \cite{ChanH20}.  The new algorithm combines the medium OPT algorithm in Section~\ref{app:medium} and the large OPT algorithm in Section~\ref{algo:large}.

Let $N$ be a fixed parameter used to control the error probability.



\ignore{
For the large OPT case (when $t>\Omega(n^{1-\epsilon})$), preprocessing takes $O(n\log n)$ time (computing separators).

If $t>\tilde{O}(n^{1/2+\epsilon})$, we can use the algorithm for large OPT in Sec.\ \ref{algo:large} to estimate the size of the optimal solution within a constant factor, so that we can avoid binary search.



For the large OPT case, since we want to output the actual solution, we need to recurse in each cluster. Set $b=O(\frac{n}{t\log n})$ and $g=O(\frac{n^2}{t^2\log^2 n})$, i.e.\ $b,g=\log^{O(1)}n$.

The additive error of the computed solution satisfies the following recurrence: $S(n)=\frac{n}{bg}\cdot S(bg)+\frac{t}{\log n}$, which gives $S(n)=O(t)$. We don't need to recurse?

The running time satisfies the following recurrence: $T(n)=\frac{n}{bg}\cdot T(bg)+O(n\log n)$, which gives $T(n)=O(n\log n)$.
}

\newcommand{\ttt}{t'}

\subparagraph*{Case 1: $\OPT \le n^{5/6+\delta}$.}
Here, we modify the medium OPT algorithm in Section~\ref{app:medium}.
We first guess a value $\ttt\in [\OPT/n^\delta,\OPT]$; a constant ($O(1/\delta)$) number of guesses suffices.
The preprocessing algorithm will use this parameter $\ttt$.

In the static setting, our old version of Lemma~\ref{lem:t0} for constructing $T_0$ suffices and takes $O(n\log n)$ time~\cite{ChanH20}.  Specifically, we construct a set $T_0\subseteq S$ of $O(\ttt)$ halfspaces so that
after removing the points covered by $T_0$, every halfspace of $S$ contains at most $O(\frac{n}{\ttt})$ points.
Since we are in the static setting, we can explicitly remove the points covered by $T_0$.
The shallow versions of the partition trees~\cite{matousek1992reporting}
can be preprocessed in $O(n\log n)$ time.
Parts of the algorithm can be simplified: 
there is no need to divide into phases, and no need for the extra counters $c^\#_v$ and the extra point set~$E$ during the MWU algorithm.
The MWU algorithm then runs in
$\OO(t\cdot (\frac n\ttt)^{2/3+O(\delta)})=\OO(t^{1/3}n^{2/3+O(\delta)})=O(n^{17/18+O(\delta)})$ time.  As we need to run the MWU algorithm for all guesses $t$ that are powers of 2 (up to $n^{5/6+\delta}$), the running time of the MWU algorithm increases by a logarithmic factor but remains $O(n^{17/18+O(\delta)})$, excluding preprocessing.


To bound the error probability by $O(\frac1N)$, we can re-run the MWU algorithm $O(\log N)$ times.  The total running time including preprocessing is $O(n\log n + n^{17/18+O(\delta)}\log N)$.

\subparagraph*{Case 2: $\OPT > n^{5/6+\delta}$.}
Here, we modify the large OPT algorithm in Section~\ref{algo:large}.
In the static setting, there is no need for the halfspace range reporting and triangle range searching structures.
We first generate the conflict lists of the cells
(the lists of halfspaces crossing the cells) in $\VD(R)$ in $O(n\log n)$ expected time~\cite{ClarksonS89}.
We verify that indeed every conflict list has size $O(b\log n)$; if not, we  restart, with $O(1)$ expected number of trials till success.
For each point $p\in X$, we locate the cell in $\VD(R)$ which contains $p$ in the $xy$-projection, by planar point location in $O(n\log n)$ time~\cite{BerBOOK}.

We refine $\VD(R)$ before applying the planar separator theorem: for each cell in $\VD(R)$, we subdivide into subcells each containing at most $b$ points of $X$ in the $xy$-projection.
The number of extra cuts is $O(\frac nb)$, and so the new decomposition still has $O(\frac nb)$ cells.

(As noted before, when there are both upper and lower halfspaces, we replace the vertical decomposition with a ``star'' triangulation and replace orthogonal $xy$-projection with a perspective projection.)

The original large OPT algorithm takes a random sample of the clusters to approximate the value of the optimal solution.
To compute
an actual solution, we instead use recursion in \emph{every} cluster.  

Recall that the number of clusters is $O(n/(bg))$.
For each cluster $\gamma$, the number of halfspaces in $S_\gamma$ is $\OO(bg)$,
and the number of points in $X_\gamma$ is also $\OO(bg)$, because of the above refinement of $\VD(R)$.
Recall that after removing halfspaces in $S_B$,
the $S_\gamma$'s become disjoint; and after removing points covered by $R\cup S_B$,
the sum of the optimal values in the subproblems  $\OPT_\gamma$  is upper-bounded by $\OPT$.

We set $b=n^{1/6}$ and $g=n^{1/3}$ so that the additive error is $\OO(\frac nb + \frac n{\sqrt{g}}) =\OO(n^{5/6})\le O(\OPT/\log n)$.

\subparagraph*{Analysis.}
The worst-case expected running time of the overall algorithm for input size $n$ satisfies the  recurrence $T(n)\leq \sum_i T(n_i)+O(n\log n + n^{17/18+O(\delta)}\log N)$, for some $n_i$'s with $\sum_i n_i\leq n$ and
$\max_i n_i = \tilde{O}(\sqrt{n})$. 
This recurrence solves to
$T(n) = O(n\log n + n\log N)$.
(Roughly, the  reason is that $\log n$ forms a geometric progression as we
descend downward in the recursion tree.)


Let $E(n,x)$ be
the worst-case additive error of the computed solution for an input of size $n$ and optimal value $x$.  In Case~1, the additive error is $O(x)$.
In general, we 
have the recurrence $E(n,x) \le \max\{O(x),\, \sum_i E(n_i,x_i)+O(\frac{x}{\log n})\}$, for some $n_i$'s and $x_i$'s with $\sum_i n_i\leq n$, and $\max_i n_i=\tilde{O}(\sqrt{n})$, and $\sum_i x_i\leq x$. This recurrence solves to $E(n,x)=O(x)$.
(Roughly, the  reason is that $\frac 1{\log n}$ forms a geometric progression as we
descend downward in the recursion tree.)
Thus, the algorithm yields an $O(1)$-approximation.

The total error probability over the entire recursion is bounded by $O(\frac nN)$.  We set $N=n^c$ for the global input size $n$ and an arbitrarily large constant $c$.

\begin{theorem}
Given $O(n)$ halfspaces and $O(n)$ points in 3D, there exists a randomized $O(1)$-approximation algorithm for the set cover problem that runs in $O(n\log n)$ expected time and is correct w.h.p.
\end{theorem}

It is possible to modify the analysis to get $O(n\log n)$ worst-case time instead of expected.
In Case~2, the conflict lists have size $O(b\log n)$ w.h.p.\ and the construction time is actually $O(n\log n)$ w.h.p., at the root of the recursion.  In subsequent levels of the recursion, we can apply a Chernoff bound to get a high-probability bound on the total running time.

\bibliographystyle{plainurl}
\bibliography{references}

\begin{thebibliography}{10}

\bibitem{AdamaszekHW19}
Anna Adamaszek, Sariel Har{-}Peled, and Andreas Wiese.
\newblock Approximation schemes for independent set and sparse subsets of
  polygons.
\newblock {\em Journal of the {ACM}}, 66(4):29:1--29:40, 2019.
\newblock \href {http://dx.doi.org/10.1145/3326122}
  {\path{doi:10.1145/3326122}}.

\bibitem{AfshaniC09}
Peyman Afshani and Timothy~M. Chan.
\newblock On approximate range counting and depth.
\newblock {\em Discrete {\&} Computational Geometry}, 42(1):3--21, 2009.
\newblock URL: \url{https://doi.org/10.1007/s00454-009-9177-z}, \href
  {http://dx.doi.org/10.1007/s00454-009-9177-z}
  {\path{doi:10.1007/s00454-009-9177-z}}.

\bibitem{agarwal2020dynamic}
Pankaj~K. Agarwal, Hsien{-}Chih Chang, Subhash Suri, Allen Xiao, and Jie Xue.
\newblock Dynamic geometric set cover and hitting set.
\newblock In {\em Proceedings of the 36th Symposium on Computational Geometry
  (SoCG)}, volume 164, pages 2:1--2:15, 2020.
\newblock URL: \url{https://doi.org/10.4230/LIPIcs.SoCG.2020.2}, \href
  {http://dx.doi.org/10.4230/LIPIcs.SoCG.2020.2}
  {\path{doi:10.4230/LIPIcs.SoCG.2020.2}}.

\bibitem{AgaEriSURV}
Pankaj~K. Agarwal and Jeff Erickson.
\newblock Geometric range searching and its relatives.
\newblock In B.~Chazelle, J.~E. Goodman, and R.~Pollack, editors, {\em Advances
  in Discrete and Computational Geometry}, pages 1--56. AMS Press, 1999.
\newblock URL: \url{http://jeffe.cs.illinois.edu/pubs/survey.html}.

\bibitem{agarwal2014near}
Pankaj~K. Agarwal and Jiangwei Pan.
\newblock Near-linear algorithms for geometric hitting sets and set covers.
\newblock {\em Discrete \& Computational Geometry}, 63(2):460--482, 2020.
\newblock Preliminary version in SoCG'14.
\newblock \href {http://dx.doi.org/10.1007/s00454-019-00099-6}
  {\path{doi:10.1007/s00454-019-00099-6}}.

\bibitem{AronovH08}
Boris Aronov and Sariel Har{-}Peled.
\newblock On approximating the depth and related problems.
\newblock {\em {SIAM} Journal on Computing}, 38(3):899--921, 2008.
\newblock \href {http://dx.doi.org/10.1137/060669474}
  {\path{doi:10.1137/060669474}}.

\bibitem{DBLP:journals/jacm/AryaMNSW98}
Sunil Arya, David~M. Mount, Nathan~S. Netanyahu, Ruth Silverman, and Angela~Y.
  Wu.
\newblock An optimal algorithm for approximate nearest neighbor searching in
  fixed dimensions.
\newblock {\em Journal of the {ACM}}, 45(6):891--923, 1998.
\newblock URL: \url{https://doi.org/10.1145/293347.293348}, \href
  {http://dx.doi.org/10.1145/293347.293348} {\path{doi:10.1145/293347.293348}}.

\bibitem{bronnimann1995almost}
Herv{\'e} Br{\"o}nnimann and Michael~T. Goodrich.
\newblock Almost optimal set covers in finite {VC}-dimension.
\newblock {\em Discrete \& Computational Geometry}, 14(4):463--479, 1995.

\bibitem{Chan00SICOMP}
Timothy~M. Chan.
\newblock Random sampling, halfspace range reporting, and construction of
  {$(\le k)$}-levels in three dimensions.
\newblock {\em {SIAM} Journal on Computing}, 30(2):561--575, 2000.
\newblock \href {http://dx.doi.org/10.1137/S0097539798349188}
  {\path{doi:10.1137/S0097539798349188}}.

\bibitem{Chan03sicomp}
Timothy~M. Chan.
\newblock Semi-online maintenance of geometric optima and measures.
\newblock {\em {SIAM} Journal on Computing}, 32(3):700--716, 2003.
\newblock URL: \url{https://doi.org/10.1137/S0097539702404389}, \href
  {http://dx.doi.org/10.1137/S0097539702404389}
  {\path{doi:10.1137/S0097539702404389}}.

\bibitem{Chan10JACM}
Timothy~M. Chan.
\newblock A dynamic data structure for 3-d convex hulls and 2-d nearest
  neighbor queries.
\newblock {\em Journal of the {ACM}}, 57(3):16:1--16:15, 2010.
\newblock \href {http://dx.doi.org/10.1145/1706591.1706596}
  {\path{doi:10.1145/1706591.1706596}}.

\bibitem{Chan12DCG}
Timothy~M. Chan.
\newblock Optimal partition trees.
\newblock {\em Discrete \& Computational Geometry}, 47(4):661--690, 2012.
\newblock \href {http://dx.doi.org/10.1007/s00454-012-9410-z}
  {\path{doi:10.1007/s00454-012-9410-z}}.

\bibitem{Chan12IJCGA}
Timothy~M. Chan.
\newblock Three problems about dynamic convex hulls.
\newblock {\em International Journal of Computational Geometry \&
  Applications}, 22(4):341--364, 2012.
\newblock \href {http://dx.doi.org/10.1142/S0218195912600096}
  {\path{doi:10.1142/S0218195912600096}}.

\bibitem{ChanH20}
Timothy~M. Chan and Qizheng He.
\newblock Faster approximation algorithms for geometric set cover.
\newblock In {\em Proceedings of the 36th Symposium on Computational Geometry
  (SoCG)}, volume 164, pages 27:1--27:14, 2020.
\newblock URL: \url{https://doi.org/10.4230/LIPIcs.SoCG.2020.27}, \href
  {http://dx.doi.org/10.4230/LIPIcs.SoCG.2020.27}
  {\path{doi:10.4230/LIPIcs.SoCG.2020.27}}.

\bibitem{chan2016optimal}
Timothy~M. Chan and Konstantinos Tsakalidis.
\newblock Optimal deterministic algorithms for 2-d and 3-d shallow cuttings.
\newblock {\em Discrete \& Computational Geometry}, 56(4):866--881, 2016.

\bibitem{Clarkson87}
Kenneth~L. Clarkson.
\newblock New applications of random sampling in computational geometry.
\newblock {\em Discrete \& Computational Geometry}, 2:195--222, 1987.
\newblock \href {http://dx.doi.org/10.1007/BF02187879}
  {\path{doi:10.1007/BF02187879}}.

\bibitem{clarkson1993algorithms}
Kenneth~L. Clarkson.
\newblock Algorithms for polytope covering and approximation.
\newblock In {\em Workshop on Algorithms and Data Structures}, pages 246--252,
  1993.

\bibitem{ClarksonS89}
Kenneth~L. Clarkson and Peter~W. Shor.
\newblock Application of random sampling in computational geometry, {II}.
\newblock {\em Discrete \& Computational Geometry}, 4:387--421, 1989.
\newblock \href {http://dx.doi.org/10.1007/BF02187740}
  {\path{doi:10.1007/BF02187740}}.

\bibitem{clarkson2007improved}
Kenneth~L. Clarkson and Kasturi Varadarajan.
\newblock Improved approximation algorithms for geometric set cover.
\newblock {\em Discrete \& Computational Geometry}, 37(1):43--58, 2007.

\bibitem{CzumajEFMNRS05}
Artur Czumaj, Funda Erg{\"{u}}n, Lance Fortnow, Avner Magen, Ilan Newman,
  Ronitt Rubinfeld, and Christian Sohler.
\newblock Approximating the weight of the {E}uclidean minimum spanning tree in
  sublinear time.
\newblock {\em {SIAM} Journal on Computing}, 35(1):91--109, 2005.
\newblock \href {http://dx.doi.org/10.1137/S0097539703435297}
  {\path{doi:10.1137/S0097539703435297}}.

\bibitem{BerBOOK}
Mark de~Berg, Otfried Cheong, Marc~J. van Kreveld, and Mark~H. Overmars.
\newblock {\em Computational Geometry: Algorithms and Applications}.
\newblock Springer, 3rd edition, 2008.

\bibitem{federickson1987fast}
Greg~N. Frederickson.
\newblock Fast algorithms for shortest paths in planar graphs, with
  applications.
\newblock {\em SIAM Journal on Computing}, 16(6):1004--1022, 1987.

\bibitem{DBLP:conf/compgeom/Henzinger0W20}
Monika Henzinger, Stefan Neumann, and Andreas Wiese.
\newblock Dynamic approximate maximum independent set of intervals, hypercubes
  and hyperrectangles.
\newblock In {\em Proceedings of the 36th Symposium on Computational Geometry
  (SoCG)}, volume 164, pages 51:1--51:14, 2020.
\newblock URL: \url{https://doi.org/10.4230/LIPIcs.SoCG.2020.51}, \href
  {http://dx.doi.org/10.4230/LIPIcs.SoCG.2020.51}
  {\path{doi:10.4230/LIPIcs.SoCG.2020.51}}.

\bibitem{LiptonT80}
Richard~J. Lipton and Robert~Endre Tarjan.
\newblock Applications of a planar separator theorem.
\newblock {\em {SIAM} Journal on Computing}, 9(3):615--627, 1980.
\newblock \href {http://dx.doi.org/10.1137/0209046}
  {\path{doi:10.1137/0209046}}.

\bibitem{LiuPL11}
Chih{-}Hung Liu, Evanthia Papadopoulou, and D.~T. Lee.
\newblock An output-sensitive approach for the {$L_1$/$L_\infty$}
  $k$-nearest-neighbor {V}oronoi diagram.
\newblock In {\em Proceedings of the 19th Annual European Symposium on
  Algorithms (ESA)}, volume 6942 of {\em Lecture Notes in Computer Science},
  pages 70--81. Springer, 2011.
\newblock \href {http://dx.doi.org/10.1007/978-3-642-23719-5\_7}
  {\path{doi:10.1007/978-3-642-23719-5\_7}}.

\bibitem{matouvsek1992efficient}
Ji{\v{r}}{\'\i} Matou{\v{s}}ek.
\newblock Efficient partition trees.
\newblock {\em Discrete \& Computational Geometry}, 8(3):315--334, 1992.

\bibitem{matousek1992reporting}
Ji{\v{r}}{\'\i} Matou{\v{s}}ek.
\newblock Reporting points in halfspaces.
\newblock {\em Computational Geometry}, 2(3):169--186, 1992.

\bibitem{matousek93DCG}
Ji{\v{r}}{\'\i} Matou{\v{s}}ek.
\newblock Range searching with efficient hierarchical cutting.
\newblock {\em Discrete \& Computational Geometry}, 10:157--182, 1993.
\newblock \href {http://dx.doi.org/10.1007/BF02573972}
  {\path{doi:10.1007/BF02573972}}.

\bibitem{MustafaRR15}
Nabil~H. Mustafa, Rajiv Raman, and Saurabh Ray.
\newblock Quasi-polynomial time approximation scheme for weighted geometric set
  cover on pseudodisks and halfspaces.
\newblock {\em {SIAM} Journal on Computing}, 44(6):1650--1669, 2015.
\newblock URL: \url{https://doi.org/10.1137/14099317X}, \href
  {http://dx.doi.org/10.1137/14099317X} {\path{doi:10.1137/14099317X}}.

\bibitem{mustafa2009ptas}
Nabil~H. Mustafa and Saurabh Ray.
\newblock Improved results on geometric hitting set problems.
\newblock {\em Discrete \& Computational Geometry}, 44(4):883--895, 2010.
\newblock \href {http://dx.doi.org/10.1007/s00454-010-9285-9}
  {\path{doi:10.1007/s00454-010-9285-9}}.

\bibitem{Sharir91}
Micha Sharir.
\newblock On {$k$}-sets in arrangement of curves and surfaces.
\newblock {\em Discrete \& Computational Geometry}, 6:593--613, 1991.
\newblock \href {http://dx.doi.org/10.1007/BF02574706}
  {\path{doi:10.1007/BF02574706}}.

\end{thebibliography}

\end{document}